\documentclass[11pt]{article}
\usepackage[margin=2.54cm]{geometry}
\usepackage[dvipsnames,usenames]{color}
\usepackage{amsfonts,amsmath,amssymb,amsthm,mathtools}
\usepackage{paralist}
\usepackage{algorithmic,algorithm}
\usepackage{bm}
\usepackage{xspace}
\usepackage[pagebackref,letterpaper=true,colorlinks=true,pdfpagemode=none,urlcolor=blue,linkcolor=blue,citecolor=BrickRed,pdfstartview=FitH]{hyperref}
\usepackage{xspace,prettyref}
\usepackage{color}
\newcommand{\ignore}[1]{}

\newtheorem{theorem}{Theorem}[section]
\newtheorem{lemma}[theorem]{Lemma}
\newtheorem{claim}[theorem]{Claim}
\newtheorem{proposition}[theorem]{Proposition}

\newtheorem{definition}[theorem]{Definition}

\def\eps{\varepsilon}

        {\hspace*{\fill}$\Box$\par}

\def\R{{\mathbb R}}

\def\x{{\bf x}}

\newcommand{\bK}{{\bf K}}

\newcommand{\cA}{{\cal A}}

\newcommand{\cD}{{\cal D}}

\newcommand{\cF}{{\cal F}}

\newcommand{\cM}{{\cal M}}

\newcommand{\NN}{\mathbb{N}}

\newcommand{\acc}{\hbox{acc}}

\newcommand{\rej}{\hbox{rej}}

\newcommand{\Sec}[1]{\hyperref[sec:#1]{\S\ref*{sec:#1}}} %section
\newcommand{\Eqn}[1]{\hyperref[eqn:#1]{(\ref*{eqn:#1})}} %equation
\newcommand{\Clm}[1]{\hyperref[clm:#1]{Claim~\ref*{clm:#1}}} %claim
\newcommand{\Fig}[1]{\hyperref[fig:#1]{Figure~\ref*{fig:#1}}} %figure
\newcommand{\Tab}[1]{\hyperref[tab:#1]{Table~\ref*{tab:#1}}} %table
\newcommand{\Thm}[1]{\hyperref[thm:#1]{Theorem~\ref*{thm:#1}}} %theorem
\newcommand{\Lem}[1]{\hyperref[lem:#1]{Lemma~\ref*{lem:#1}}} %lemma
\newcommand{\Prop}[1]{\hyperref[prop:#1]{Proposition~\ref*{prop:#1}}} %property
\newcommand{\Cor}[1]{\hyperref[cor:#1]{Corollary~\ref*{cor:#1}}} %corollary
\newcommand{\Def}[1]{\hyperref[def:#1]{Definition~\ref*{def:#1}}} %definition
\newcommand{\Alg}[1]{\hyperref[alg:#1]{Algorithm~\ref*{alg:#1}}} %algorithm
\newcommand{\Ex}[1]{\hyperref[ex:#1]{Example~\ref*{ex:#1}}} %example

\begin{document}

\author{Deeparnab Chakrabarty
  \\\\ 
  Microsoft Research India\\
  {\tt dechakr@microsoft.com} \\
\and C. Seshadhri \\\\ 
Sandia National Labs, Livermore\thanks{Sandia National Laboratories is a multi-program laboratory managed and operated by Sandia Corporation, a wholly owned subsidiary of Lockheed Martin Corporation, for the U.S. Department of Energy's National Nuclear Security Administration under contract DE-AC04-94AL85000.} \\
{\tt scomand@sandia.gov}}

\title{An optimal lower bound for monotonicity testing over hypergrids}

%\author{D. Chakrabarty \\\\
%{Microsoft Research India, 9 Lavelle Road, Bangalore, 560001}\\
%\tt{dechakr@microsoft.com}
%\and C. Seshadhri \\\\
%{Sandia National Labs, Livermore}
%\thanks{Sandia National Laboratories is a multi-program laboratory managed and operated by Sandia Corporation, a wholly owned subsidiary of Lockheed Martin Corporation, for the U.S. Department of Energy's National Nuclear Security Administration under contract DE-AC04-94AL85000.} \\
%\tt{scomand@sandia.gov}
%}

\date{}
\maketitle
\def\D{{\mathbf D}}
\def\R{{\mathbf R}}

\begin{abstract} For positive integers $n, d$, consider the hypergrid $[n]^d$ with the coordinate-wise product partial ordering denoted by $\prec$. 
A function $f: [n]^d \mapsto \NN$ is monotone if $\forall x \prec y$, $f(x) \leq f(y)$.
A function $f$ is $\eps$-far from monotone if at least an $\eps$-fraction of values must be changed to make
$f$ monotone. Given a parameter $\eps$, a \emph{monotonicity tester} must distinguish with high probability a monotone function from one that is $\eps$-far.

We prove that any (adaptive, two-sided) monotonicity tester for functions $f:[n]^d \mapsto \NN$ must make
$\Omega(\eps^{-1}d\log n - \eps^{-1}\log \eps^{-1})$ queries. Recent upper bounds show the existence of $O(\eps^{-1}d \log n)$
query monotonicity testers for hypergrids. This closes the question of monotonicity testing for hypergrids
over arbitrary ranges. The previous best lower bound for general hypergrids was a non-adaptive bound
of $\Omega(d \log n)$. 
%Optimal adaptive lower bounds were  for the hypercube $d=2$.
\end{abstract}

\section{Introduction} \label{sec:intro}

Given query access to a function $f: \D \mapsto \R$, the field of \emph{property testing}~\cite{RS96,GGR98} deals
with the problem of determining properties of $f$ without reading all of it.
Monotonicity testing~\cite{GGLRS00} is a classic problem in property testing. Consider a function $f: \D \mapsto \R$,
where $\D$ is some partial order given by ``$\prec$", and $\R$ is a total order. The function $f$ is monotone
if for all $x \prec y$ (in $\D$), $f(x) \leq f(y)$. The \emph{distance to monotonicity} of $f$ is the minimum
fraction of values that need to be modified to make $f$ monotone. More precisely, define the distance between functions
$d(f,g)$ as $|\{x: f(x) \neq g(x)\}|/|\D|$. Let $\cM$ be the set of all monotone functions.
Then the distance to monotonicity of $f$ is $\min_{g \in \cM} d(f,g)$.

A function is called $\eps$-far from monotone if the distance to monotonicity is at least $\eps$.
A \emph{property tester for monotonicity} is a, possibly randomized, algorithm that takes as input a distance parameter $\eps \in (0,1)$, error parameter $\delta \in [0,1]$, 
and query access to an arbitrary $f$. If $f$ is monotone, then the tester must accept with probability $>1-\delta$.
If it is $\eps$-far from monotone, then the tester rejects with probability $>1-\delta$. (If neither, then the tester
is allowed to do anything.) The aim is to design a property tester using as few queries as possible.
A tester is called \emph{one-sided} if it always accepts a monotone function. A tester is called \emph{non-adaptive}
if the queries made do not depend on the function values. The most general tester is an adaptive two-sided tester.

Monotonicity testing has a rich history and the hypergrid domain, $[n]^d$, has received special attention.
The boolean hypercube ($n=2$) and the total order ($d=1$) are special instances of hypergrids.
Following a long line of work~\cite{EKK+00, GGLRS00,DGLRRS99,LR01,FLNRRS02,AC04,E04,HK04,PRR04,ACCL04,BRW05,BBM11}, 
previous work of the authors~\cite{ChSe13} shows the existence of $O(\eps^{-1} d\log n)$-query
monotonicity testers. Our result is a matching adaptive lower bound that is optimal in all parameters
(for unbounded range functions).
This closes the
question of monotonicity testing for unbounded ranges on hypergrids.
This is also the first adaptive bound for monotonicity testing on general hypergrids.

\begin{theorem} \label{thm:main}
Any (adaptive, two-sided) monotonicity tester for functions $f:[n]^d \mapsto \NN$ requires $\Omega(\eps^{-1} d\log n - \eps^{-1}\log \eps^{-1})$ queries.
\end{theorem}

\subsection{Previous work} \label{sec:prev}

The problem of monotonicity testing was introduced by Goldreich et al~\cite{GGLRS00}, with an $O(n/\eps)$
tester for functions $f:\{0,1\}^n \mapsto \{0,1\}$. The first tester for general hypergrids was given
by Dodis et al~\cite{DGLRRS99}. The upper bound of $O(\eps^{-1} d\log n)$ for monotonicity testing was recently proven
in~\cite{ChSe13}. We refer the interested reader to the introduction of~\cite{ChSe13} for a more detailed history of previous upper bounds.

There have been numerous lower bounds for monotonicity testing. We begin by summarizing the state of the art.
The known adaptive lower bounds are $\Omega(\log n)$ for the total order $[n]$ by Fischer~\cite{E04},
and $\Omega(d/\eps)$ for the boolean hypercube $\{0,1\}^d$ by Brody~\cite{Br13}. For general hypergrids,
Blais, Raskhodnikova, and Yaroslavtsev~\cite{BlJh+12} recently proved the first result, a non-adaptive lower bound of $\Omega(d\log n)$.
\Thm{main} is the first adaptive bound for monotonicity testing on hypergrids and is optimal (for arbitrary ranges) in all parameters.
%It also subsumes all previous bounds for special cases.

Now for the chronological documentation. The first lower bound was the non-adaptive bound of $\Omega(\log n)$ for the total order $[n]$ by Ergun et al~\cite{EKK+00}.
This was extended by Fischer~\cite{E04} to an (optimal) adaptive bound. For the hypercube domain $\{0,1\}^d$,
Fischer et al~\cite{FLNRRS02} proved the first non-adaptive lower bound of $\Omega(\sqrt{d})$. (This was proven
even for the range $\{0,1\}$.) This was improved to $\Omega(d/\eps)$ by Br\"{i}et et al~\cite{BCG+10}.
Blais, Brody, and Matulef~\cite{BBM11} gave an ingenious reduction from communication complexity
to prove an adaptive, two-sided bound of $\Omega(d)$. (Honing this reduction, Brody~\cite{Br13} improved this bound to $\Omega(d/\eps)$.)
The non-adaptive lower bounds of Blais, Raskhodnikova, and Yaroslavtsev~\cite{BlJh+12} were also achieved
through communication complexity reductions.

%\begin{table}[h] 
%\caption{Monotonicity lower bounds over hypergrids} \label{tab:lb}
%\centering % centering table
%\begin{tabular}{l l l l l} % creating 5 columns
%\hline\hline % inserting double-line
%Paper & Domain & Setting & Bound \\ [0.5ex]
%\hline % inserts single-line
%\cite{EKK+00} & $[n]$ & Non-adaptive & $\log n$ \\%[-0.6ex]
%\cite{E04} & $[n]$ & Adaptive & $\log n$ \\
%\cite{FLNRRS02} & $\{0,1\}^d$ & Non-adaptive & $\sqrt{d}$\\
%\cite{BCG+10} & $\{0,1\}^d$ & Non-adaptive & $\eps^{-1}d$\\
%\cite{BBM11} & $\{0,1\}^d$ & Adaptive & $d$\\
%\cite{Br13} & $\{0,1\}^d$ & Adaptive & $\eps^{-1}d$\\
%\cite{BlJh+12} & $[n]^d$ & Non-adaptive & $d\log n$\\
%This paper & $[n]^d$ & Adaptive & $\eps^{-1}d\log n$\\\\ [0.5ex]
%\hline % inserts single-line
%\end{tabular}
%\end{table}
%
We note that our theorem only holds when the range is $\NN$, while some previous results hold for restricted
ranges. The results of~\cite{BBM11, Br13} provide lower bounds for range $[\sqrt{d}]$.
The non-adaptive bound of~\cite{BlJh+12} holds even when the range is $[nd]$. In that sense,
the communication complexity reductions provide stronger lower bounds than our result.

\subsection{Main ideas} \label{sec:ideas}

The starting point of this work is the result of Fischer~\cite{E04}, an adaptive lower bound for monotonicity
testing for functions $f:[n] \mapsto \NN$. He shows that adaptive testers can be converted to comparison-based
testers, using Ramsey theory arguments. A comparison-based tester for $[n]$ can be easily converted to a non-adaptive
tester, for which an $\Omega(\log n)$ bound was previously known. We make a fairly simple observation. The main part of Fischer's proof
actually goes through for functions over \emph{any partial order}, so it suffices to prove lower bounds
for comparison-based testers. (The reduction to non-adaptive testers only holds for $[n]$.)

We then prove a comparison-based lower bound of $\Omega(\eps^{-1} d\log n - \eps^{-1}\log \eps^{-1})$
for the domain $[n]^d$. 
As usual, Yao's minimax lemma allows us to prove determinstic lower bounds over some distribution of functions.
The major challenge in proving (even non-adaptive) lower bounds for monotonicity is that the tester might make decisions
based on the actual values that it sees. Great care is required to construct a distribution over functions
whose monotonicity status cannot be decided by simply looking at the values.
But a comparison-based tester has no such power, and
optimal lower bounds over all parameters can be obtained with a fairly clean distribution.

\section{The reduction to comparison based testers} \label{sec:prelim}

Consider the family of functions $f: \D \mapsto \R$, where $\D$ is some partial order, and $\R\subseteq \NN$.
We will assume that $f$ always takes distinct values, so $\forall x,y, f(x) \neq f(y)$. Since we are proving
lower bounds, this is no loss of generality.

\begin{definition} \label{def:tester}
An algorithm $\cA$ is a \emph{$(t,\eps,\delta)$-monotonicity tester}
if $\cA$ has the following properties. For any $f:\D\mapsto \R$, the algorithm $\cA$ makes $t$ (possibly randomized) queries to $f$
and then outputs either ``accept" or ``reject". If $f$ is monotone, then $\cA$ accepts with probability $>1-\delta$.
If $f$ is $\eps$-far from monotone, then $\cA$ rejects with probability $>1-\delta$.
\end{definition}
\def\acc{{\tt acc}}
\def\rej{{\tt rej}}
\def\x{{\mathbf x}}
\def\a{{\mathbf a}}
%
%We will focus on functions that only take distinct values, so $f(x) \neq f(y)$ for all $x, y \in \D$. This is no loss of generality
%for proving lower bounds. Formally, any function $f$ can be transformed
%

Given a positive integer $s$, let $\D^s$ denote the collection of {\em ordered}, $s$-tupled vectors with each entry in $\D$.
We define two symbols \acc~ and \rej, and denote $\D' = \D \cup \{\acc, \rej\}$.
Any $(t,\eps,\delta)$-tester can be completely specified by the following family of functions.
For all $s\leq t$, $\x\in \D^s$, $y\in \D'$, we consider a function $p^y_\x:\R^s \mapsto [0,1]$, with the semantic that for any $\a\in \R^s$, $p^y_\x(\a)$ denotes the probability the tester queries $y$ as the $(s+1)$th query, given that the 
first $s$ queries are $\x_1,\ldots,\x_s$ and $f(\x_i) = \a_i$ for $1\leq i\leq s$. By querying $\acc,\rej$ we imply returning accept or reject. These functions satisfy the following properties.

%Let $x, y$ denote points in $D$ and $\v$ denotes natural numbers. Consider the following \emph{finite} family
%of functions:
\begin{align}
%& \forall s \leq t, \forall \x \in D^s, \forall y \in D', \ p^y_\x: \NN^s \mapsto [0,1] \nonumber \\
%& \forall s \leq t, \forall x_1, x_2, \ldots, x_s \in D, \ \acc_{x_1, x_2, \ldots, x_s}: \NN^s \mapsto [0,1], \rej_{x_1, x_2, \ldots, x_s}: \NN^s \mapsto [0,1] \nonumber \\
%& \forall s \leq t, \forall x_1, x_2, \ldots, x_s \in D, \forall v_1, v_2, \ldots, v_s \in \NN, \nonumber\\
%& \ \ \ \acc_{x_1, x_2, \ldots, x_s}(v_1, v_2, \ldots, v_s) + \rej_{x_1, x_2, \ldots, x_s}(v_1, v_2, \ldots, v_s) + \sum_{y \in D} p_{y, x_1, x_2, \ldots, x_s}(v_1, v_2, \ldots, v_s) = 1 \label{eqn:dist}\\
& \forall s \leq t,~ \forall\x\in \D^s,~ \forall \a \in \R^s,~ \sum_{y \in \D'} p^y_\x(\a) = 1 \label{eqn:dist}\\
& \forall \x\in \D^t,~ \forall y \in \D,~ \forall \a \in \R^t,~ p^y_\x(\a) = 0 \label{eqn:time}
\end{align}
\Eqn{dist} ensures the decisions of the tester at step $(s+1)$ must form a probability distribution. \Eqn{time} implies that the tester makes at most $t$ queries.
%Think of $p_{y, x_1, x_2, \ldots, x_s}(v_1, v_2, \ldots, v_s)$ as answering: ``what is the probability of outputing $y$
%as the next query, given that $f(x_1) = v_1$, $f(x_2) = v_2$, $\ldots, f(x_s) = v_s$?". (If the next query is ``acc" or ``rej",
%then the tester simply makes that decision.)
%Note that if we fix
%$x_1, \ldots, x_s$ and $y_1, \ldots, y_s$, then we must get a probability distribution on the decisions of the tester,
%which is specified by \Eqn{dist}. After $t$ queries, the tester must make decision and cannot ask for another value,
%as expressed by \Eqn{time}.

For any positive integer $s$, let $\R^{(s)}$ denote {\em unordered} sets of $\R$ of cardinality $s$. For reasons that will soon become clear, we introduce new functions as follows. For each $s$,  $\x\in \D^s$, $y\in \D'$, and {\em each permutation} $\sigma:[s]\mapsto [s]$, we associate functions $q^y_{\x,\sigma}:\R^{(s)} \mapsto [0,1]$, with the semantic
$$\textrm{ For any set $S = (a_1 < a_2 < \cdots < a_s) \in \R^{(s)}$, } ~~~~q^y_{\x,\sigma}(S) := p^y_\x(a_{\sigma(1)},\ldots,a_{\sigma(s)})$$

That is, $q^y_{\x,\sigma_s}(S)$ sorts the answers in $S$ in increasing order, permutes it according to $\sigma$, and passes the permuted ordered tuple to $p^y_\x$. Any adaptive tester can be specified by these functions. The important point to note is that they are finitely many such functions; their number is upper bounded by $(t|\D|)^{t+1}$. These $q$-functions allow us to define comparison based testers.

\begin{definition} \label{def:comp}
A monotonicity tester $\cA$ is \emph{comparison-based} if for all $s$,$\x\in \D^s,y\in\D'$, and permutations $\sigma:[s]\mapsto[s]$, the function $q^y_{\x,\sigma}$ is a constant function on $\R^{(s)}$. 
%That is, for all $A,B\in \R^{(s)}$, $q^y_{\x,\sigma}(A) = q^y_{\x,\sigma}(B)$.
In other words, the $(s+1)$th decision of the tester given that the first $s$ questions is $\x$, depends only on the {\em ordering} of the answers received, and not on the values of the answers.
\end{definition}

The following theorem is implicit in the work of Fischer~\cite{E04}.

\begin{theorem} \label{thm:fischer} Suppose there exists a $(t,\eps,\delta)$-monotonicity tester
for functions $f: \D \mapsto \NN$. Then there exists a comparison-based $(t,\eps,2\delta)$-monotonicity tester
for functions $f: \D \mapsto \NN$.
\end{theorem}

This implies that a comparison-based lower bound suffices for proving a general lower bound on monotonicity testing.
We provide a proof of the above theorem in the next section for completeness.

\subsection{Performing the reduction} \label{sec:comp}

We basically present Fischer's argument, observing that $\D$ can be any partial order. 
A monotonicity tester is called \emph{discrete} if the corresponding functions $p^y_{\x}$ can
only take values in $\{i/K \ | \ 0 \leq i \leq K\}$ for some finite $K$. Note that this implies the functions $q^y_{\x,\sigma}$ also take discrete values.

\begin{claim} \label{clm:discrete} Suppose there exists a $(t,\eps,\delta)$-monotonicity tester $\cA$
for functions $f: \D \mapsto \NN$. Then there exists a discrete $(t,\eps,2\delta)$-monotonicity tester for these functions.
\end{claim}
\def\p{\hat{p}}

\begin{proof} We do a rounding on the $p$-functions. Let $K = 100t|\D|^{t}/\delta^2$.
Start with the $p$-functions of the $(t,\eps,\delta)$-tester $\cA$.
For $y \in \D \cup \acc$, $\x\in\D^s$, $\a\in \R^s$, let 
$\p^y_{\x}(\a)$ be the largest value in $\{i/K \ | \ 0 \leq i \leq K\}$
at most $p^y_\x(\a)$. 
Set $\p^\rej_\x(\a)$ so that \Eqn{dist} is maintained. 

Note that for  $y \in \D \cup \acc$, if $p^y_\x(\a) > 10|\D|t/(\delta K)$, then 
$$ \left(1 - \frac{\delta}{10|\D|t}\right) p^y_\x(\a) ~ \leq ~\p^y_\x(\a)~ \leq ~p^y_\x(\a).$$
Furthermore, $\p^\rej_\x(\a) \geq p^\rej_\x(\a)$.
%$$ |\p^\rej_\x(\a) - p^\rej_\x(\a)| \leq \frac{|\D|}{K}$$
%

The $\p$-functions describe a new discrete tester $\cA'$ that makes at most $t$ queries. We argue that $\cA'$ is a $(t,\eps,2\delta)$-tester.
Given a function $f$ that is either monotone or $\eps$-far from monotone, 
consider a sequence of queries $x_1,\ldots,x_s$ after which $\cA$ returns a {\em correct} decision $\aleph$. 
Call such a sequence good, and let $\alpha$ denote the probability this occurs. We know that the sum of probabilities over all good query sequences is at least $(1-\delta)$. Now,
$$\alpha := p^{x_1}\cdot ~p^{x_2}_{x_1}(f(x_1)) \cdot~ p^{x_3}_{(x_1,x_2)}(f(x_1),f(x_2))\cdots \cdot~ p^{\aleph}_{(x_1,\ldots,x_s)}(f(x_1),\ldots,f(x_s))$$

Two cases arise. suppose all of the probabilities in the RHS are $\geq 10t/\delta K$. Then, the probability of this good sequence arising in $\cA'$ is at least $(1-\delta/10t)^t\alpha \geq \alpha(1-\delta/2)$.
%lies in the range $\left[(1-\delta/10|\D|t)^t \alpha, \alpha\left(1 + \frac{|\D|}{K}\frac{\delta K}{10|\D|t}\right)\right]$. 
%By choice of $K$, this range is contained in $[\alpha(1-\delta/2), \alpha(1 + \delta/2)]$. 
Otherwise, suppose some probability in the RHS is $<10t/\delta K$. Then the total probability mass on such good sequences in $\cA$ is atmost $10t/\delta K\cdot|\D|^t \leq \delta/2$.
Therefore, the probability of good sequences in $\cA'$ is at least $(1-3\delta/2)(1-\delta/2) \geq 1-2\delta$.
That is, $\cA'$ is a $(t,\eps,2\delta)$ tester.
%Consider some sequence of queries $x_1, x_2, \ldots, x_s, z$, where $z$ is some decision.
%Let $\alpha$ denote the probability of this happening for $\cA$ evaluated on $f$. Suppose every probability
%($1 \leq r \leq s$) $p_{x_r, x_1, \cdots, x_{r-1}}(f(x_1), \ldots, f(x_{r-1}))$ is at least $10t/K$.
%Then the probability of this sequence in $\cA'$ lies in the range $[(1-\delta/10t)^t \alpha, \alpha + |D|/K]$. 
%By choice of $K$, this range is $[(1-\delta/2)\alpha, \alpha + \delta/2]$.
%The total probability of all sequences $x_1, x_2, \ldots, x_s, z$ where some $p_{x_r, x_1, \cdots, x_{r-1}}(f(x_1), \ldots, f(x_{r-1})) \leq 10t/(\delta K)$
%is at most $10t|D|^t/(\delta K) \leq \delta/2$. The total probability of all accepting paths in $\cA$
%differs by at most $\delta$ of the corresponding probability in $\cA'$. Hence, the extra error
%incurred by $\cA'$ is at most $\delta$, and it is a $(t,\eps,2\delta)$-monotonicity tester.
\end{proof}

\def\col{{\tt col}}

We introduce some Ramsey theory terminology. %Consider a countable set $Z$. 
For any positive integer $i$, a {\em finite} coloring of $\NN^{(i)}$ is a function $\col_i:\NN^{(i)} \mapsto \{1,\ldots,C\}$
for some finite number $C$. An infinite set $X\subseteq \NN$ is called {\em monochromatic} w.r.t  $\col_i$ if for all sets $A,B \in X^{(i)}$, $\col_i(A) = \col_i(B)$. A {\em $k$-wise} finite coloring of $\NN$ is a collection of $k$ colorings
$\col_1,\ldots,\col_k$. (Note that each coloring is over different sized tuples.) 
An infinite set $X\subseteq \NN$ is $k$-wise monochromatic if $X$ is monochromatic w.r.t. all the $\col_i$'s.

The following is a simple variant of Ramsey's original theorem.
(We closely follow the proof of Ramsey's theorem as given in Chap V1, Theorem 4 of~\cite{Bol00}.)
%\comment{need some citations}.
%$c$-coloring of $Z^{(i)}$
%is a function $C: Z^{(i)} \mapsto \{1, 2, \ldots, c\}$. A set $S \subseteq Z$ is \emph{monochromatic} 
%if the restriction of $C$ to $S^{(i)}$ is a constant function (in other words, all tuples of $S^{(i)}$ have the same color).

%A \emph{$k$-wise $c$-coloring} of $Z$ is a collection of $k$ functions $C^i: Z^{(i)} \mapsto \{1, 2, \ldots, c\}$, $i \in [k]$.
%A $k$-wise monochromatic set $S$ is one where for each $i$, tuples of $S^{(i)}$ have the same color. Note that the color
%do not have to be same across tuples of different sizes.

\begin{theorem} \label{thm:ramsey} For any $k$-wise  finite coloring of $\NN$,
there is an infinite $k$-wise monochromatic set $X\subseteq \NN$.
\end{theorem}
%\comment{Multiset issues?? Refer?}
\begin{proof} We proceed by induction on $k$. If $k = 1$, then this is trivially true; let $X$ be
the maximum color class. Since the coloring is finite, $X$ is infinite.
We will now iteratively construct an infinite set of $\NN$ via induction.

Start with $a_0$ being the minimum element in $\NN$. Consider a $(k-1)$-wise coloring of $(\NN \setminus \{a_0\})$
$\col'_1,\ldots,\col'_{k-1}$, where $\col'_i(S) := \col_{i+1}(S\cup a_0)$.
%given by $C^i_1$ ($1 \leq i \leq k-1$), where $C^i_1(S) = C^{i+1}(S \cup \{a_0\})$. 
By the induction hypothesis, there exists an infinite $(k-1)$-wise monochromatic
set $A_0\subseteq \NN\setminus\{a_0\}$ with respect to coloring $\col'_i$s. 
That is, for $1\leq i\leq k$, and any set $S,T\subseteq A_0$ with $|S|=|T|=i-1$, we have 
$\col_i(a_0\cup S)=\col_i(a_0\cup T) = C^0_i$, say.
%all $i$-tuples of $a_0 \cup A_0$ which contain $a_0$, have the same color.
%That means, for each $i \in [k]$, all $i$-tuples
%of $\{a_0\} \cup A_0$ have the same color.
\def\C{{\mathbf C}}
Denote the collection of these colors as a vector $\C_0 = (C^0_1, C^0_2, \ldots, C^0_k)$.

 Subsequently, let $a_1$ be the minimum element in $A_0$, and consider the $(k-1)$-wise coloring $\col'$ of $(A_0 \setminus \{a_1\})$ where $\col'_i(S) = \col_{i+1}(S \cup \{a_1\})$ for $S\subseteq A_0\setminus\{a_1\}$.
Again, the induction hypothesis yields an infinite $(k-1)$-wise monochromatic set $A_1$ as before, and similarly the vector $\C_1$. 
Continuing this procedure, 
we get an infinite sequence $a_0, a_1, a_2, \ldots$ of natural numbers, an infinite sequence
of vectors of $k$ colors $\C_0, \C_1, \ldots$, and an infinite nested sequence of infinite sets $A_0 \supset A_1 \supset A_2 \ldots$.
Every $A_r$ contains $a_s, \forall s > r$ and by construction, any set $(\{a_r\} \cup S)$, $S\subseteq A_r$, $|S|=i-1$, 
has color $C^i_r$. Since there are only finitely many colors, some vector of colors occurs infinitely often as $\C_{r_1}, \C_{r_2}, \ldots$. The corresponding
infinite sequence of elements $a_{r_1}, a_{r_2}, \ldots$ is $k$-wise monochromatic.
\end{proof}

\begin{proof} (of \Thm{fischer}) 
Suppose there exists a $(t,\eps,\delta)$-tester for functions $f:\D\mapsto \NN$. We need to show there is a comparison-based $(t,\eps,2\delta)$-tester for such functions.

By \Clm{discrete}, there is a discrete $(t,\eps,2\delta)$-tester $\cA$. Equivalently, we have the functions $q^y_{\x,\sigma}$ as described in the previous section.
We now describe a $t$-wise finite coloring of $\NN$. 
%It suffices to describe $\col_s$ for $1\leq s\leq t$. 
Consider $s \in [t]$.
Given a set $A\subseteq \NN^{(s)}$, $\col_s(A)$ is a vector indexed by $(y,\x,\sigma)$, where $y\in D'$, $\x\in D^s$, and $\sigma$ is a $s$-permutation, whose entry is $q^y_{\x,\sigma}(A)$.
The domain is finite, so the number of dimensions is finite.
Since the tester is discrete, the number of possible colors entries is finite.  
Applying \Thm{ramsey}, we know the existence of a $t$-wise monochromatic infinite set $\R\subseteq \NN$. We have the property that for any $y,\x,\sigma$, and any two sets $A,B \in \R^{(s)}$, we have $q^y_{\x,\sigma}(A) = q^y_{\x,\sigma}(B)$. That is, the algorithm $\cA$ is a comparison based tester for functions with range $\R$.

Consider the strictly monotone map $\phi: \NN \mapsto \R$, where $\phi(b)$ is the $b$th element of $\R$  in sorted order. Now given any function $f:\D\mapsto \NN$, consider the function $\phi\circ f:\D\mapsto \R$. 
Consider an algorithm $\cA'$ which on input $f$ runs $\cA$ on $\phi\circ f$. More precisely, whenever $\cA$ queries a point $x$, it gets answer $\phi\circ f(x)$.  Observe that if $f$ is monotone (or $\eps$-far from monotone), then so is $\phi\circ f$, and therefore, the algorithm $\cA'$ is a $(t,\eps,2\delta)$-tester of $\phi\circ f$. Since the range of $\phi\circ f$ is $\R$,  
$\cA'$ is comparison-based.
%Since $\phi$ is strictly monotone, the behavior of a comparison-based tester is same for both $f$ and $\phi\circ f$, and so $\cA$ is also a $(t,\eps,2\delta)$-comparison based tester for $f$.\comment{check this}
\end{proof}%
%
%
%Start with a $(t,\eps,\delta)$-tester, \Clm{discrete}
%gives a discrete $(t,\eps,2\delta)$-tester. We now construct a $t$-wise coloring of $\NN$. 
%For some $i \in [k]$ and set $S \in \NN^{(i)}$, consider the list of all values (with some arbitrary order over \emph{all} the $p$-functions)
%$p_{y,x_1, x_2, \ldots, x_i, \sigma_i}(S)$. This list has finitely many possible values, since it has finite many entries from a finite set.
%(We use the fact that the tester is discrete.) This list is the ``color" of $S$. 
%In this way, we get a $t$-wise finite coloring of $\NN$.
%
%Applying the modified Ramsey's theorem of \Thm{ramsey}, there is some infinite set $t$-wise monochromatic $S \subset \NN$.
%Hence, all function $p_{y,x_1, x_2, \ldots, x_i, \sigma_i}$ are constants when restricted to subsets of $S$.
%Consider the strictly monotone map $\phi: \NN \mapsto S$, where $\phi(b)$ is the $b$th element of $S$ (in sorted order).
%Observe that for any function $f$, $\cA$ is a comparison-based tester for $\phi \circ f$. By \Clm{map},
%we have a $(t,\eps,2\delta)$-comparison based tester.
%
%
%If we had an input $f: D \mapsto \NN$ that took values only in $S$, then $\cA$ is comparison-based on $f$.
%
%

\section{Lower bounds} \label{sec:dist}

\def\f{\widetilde{f}}
\def\v{{\tt val}}
We assume that $n$ is a power of $2$, set $\ell:=\log_2 n$, and think of $[n]$ as $\{0,1,\ldots,n-1\}$. 
For any number $0 \leq z < n$, we think of the binary representation as $z$ as an $\ell$-bit vector $(z_1, z_2, \ldots, z_\ell)$,
where $z_1$ is the least significant bit.

%Also let's assume $n$ is a power of $2$ and $\ell:=\log_2 n$. 
Consider the following canonical, one-to-one mapping $\phi: [n]^d \mapsto \{0,1\}^{d\ell}$.
For any $\vec{y} = (y_1, y_2, \ldots, y_d) \in [n]^d$, we concatenate their binary representations in order 
to get a $d\ell$-bit vector $\phi(\vec{y})$. Hence, we can transform a function $f:\{0,1\}^{d\ell}\mapsto \NN$
into a function $\f:[n]^d \mapsto \NN$ by defining $\f(\vec{y}) := f(\phi(\vec{y}))$.
%
%
%We describe a canonical, one-to-one mapping $\Psi$ between functions $f:\{0,1\}^{d\ell}\mapsto \NN$ and functions $f:[n]^d\mapsto\NN$. 
%%Given $f$, we map it to the unique function $\f:\{0,1\}^{d\ell}\mapsto \NN$ as follows.
%%For any $x\in \{0,1\}^{d\ell}$, let $y_i$ denote the number whose binary representation is $(x_{(i-1)\ell+1}, \cdots, x_{i\ell})$. Note that $0\leq y_i\leq n-1$. 
%For any sequence $y_1, y_2, \ldots, y_d$ of numbers in $[n]$, let $\bf{y}$ denote $d\ell$-bit vector obtained
%by concatenating the binary representations of $y_1, y_2, \ldots, y_d$.
%We let $(\Psi f)(y_1, y_2, \ldots, y_d) := f(\bf{y})$. 

We will now describe a distribution of functions over the boolean hypercube with equal mass on monotone and $\eps$-far from monotone functions. The key property is that for a function drawn from this distribution, any deterministic comparison based algorithm errs in classifying it with non-trivial probability. 
This property will be used in conjunction with the above mapping to get our final lower bound.

\subsection{The hard distribution} \label{sec:hard}

We focus on functions $f:\{0,1\}^m \mapsto \NN$. (Eventually, we set $m = d\ell$.)
%We now describe the distribution.
Given any $x \in \{0,1\}^m$, we let $\v(x) :=  \sum_{i = 1}^{m} 2^{i-1} x_i$ denote the number for which $x$ is the binary representation. Here, $x_1$ denotes the least significant bit of $x$. 

For convenience, we let $\eps$ be a power of $1/2$. %Let $\eps = 1/2^r$.
%Consider the sequence of intervals of natural numbers: 
%$$\{0,1,\ldots,2\eps 2^m-1\},\{2\eps 2^m,\ldots,  4\eps 2^m-1\},\{4\eps 2^m,\ldots, 6\eps 2^m-1\},\ldots, \{(1 - 2\eps)2^m, \ldots, 2^m-1\}.$$ 
For $k \in \{1,\ldots,\frac{1}{2\eps}\}$, we let  
%$$S_k := \{x: \v(x) \in \{2(k-1)\eps 2^m, 2(k-1)\eps 2^m+1,\ldots, 2k\eps 2^m-1\}~~~\}.$$
$$S_k := \{x: \v(x) \in [2(k-1)\eps 2^m, 2k\eps 2^m-1)~~\}.$$
Note that $S_k$'s partition the hypercube, with each $|S_k| = \eps2^{m+1}$.
In fact, each $S_k$ is a subhypercube of dimension $m' := m+1-\log(1/\eps)$, with the minimal element having all zeros in the $m'$ least significant bits, and the maximal element having all ones in those.

We describe a distribution $\cF_{m,\eps}$ on functions. 
The support of $\cF_{m,\eps}$ consists $f(x) = 2\v(x)$,
and $\frac{m'}{2\eps}$ functions indexed as $g_{j,k}$ with $j\in [m']$ and $k\in[\frac{1}{2\eps}]$, defined as follows.
$$ g_{j,k}(x) = 
\left\{ 
\begin{array} {l l}
2\v(x) - 2^{j} - 1& \quad \textrm{if $x_j = 1$ and $x \in S_k$}\\
2\v(x) & \quad \textrm{otherwise} 
\end{array} \right.
$$

The distribution $\cF_{m,\eps}$ puts probability mass $1/2$ on the function $f=2\v$ and $\frac{\eps}{m'}$ on each of the $g_{j,k}$'s. 
All these functions take distinct values on their domain.
Note that $2\v$ induces a total order on $\{0,1\}^m$.\smallskip

\noindent
\textbf{The distinguishing problem:} Given query access to a random function $f$ from $\cF_{m,\eps}$, we want a deterministic comparison-based algorithm
that declares that $f = 2\v(x)$ or $f \neq 2\v(x)$. We refer to any such algorithm as a \emph{distinguisher}. 
Naturally, we say that the distinguisher errs on $f$ if it's declaration is wrong.
Our main lemma is the following.

\begin{lemma} \label{lem:main-hard} Any deterministic comparison-based distinguisher that makes less than $\frac{m'}{8\eps}$ queries errs
with probability at least $1/8$.
\end{lemma}

\noindent
The following proposition allows us to focus on {\em non-adaptive} comparison based testers. 
%Any (possibly adaptive) comparison-based distinguisher is better off rejecting if the answers to a set of queries $\x$ is not consistent with $\v(x_i)$'s. Therefore, the tester might as well assume the answers are consistent and make them up front.  

\begin{proposition}	\label{prop:comp}
Given any deterministic comparison-based distinguisher $\cA$ for $\cF_{m,\eps}$ that makes at most $t$ queries, there exists a deterministic {\em non-adaptive} comparison-based distinguisher $\cA'$ making at most $t$ queries whose probability of error on $\cF_{m,\eps}$ is at most that of $\cA$.
%Given any deterministic comparison based tester $\cA$ which makes at most $t$ queries and errs with probability $\delta$ on functions drawn from $\cF_m$, there exists a deterministic {\em non-adaptive} comparison-based tester which makes at most $t$ queries and errs with probability $\leq \delta$ on functions drawn from $\cF_m$.
\end{proposition}
\def\ans{{\tt ans}}
\begin{proof}
We represent $\cA$ as a comparison tree. For any path in $\cA$, the total number of distinct domain points involved in comparisons is at most $t$.
Note that $2\v(x)$ is a total order, since for any $x,y$ either $\v(x) < \v(y)$ or vice versa.
For any comparison in $\cA$, there is an outcome inconsistent with this ordering. (An outcome ``$f(x) < f(y)$"
where $\v(x) > \v(y)$ is inconsistent with the total order.)
We construct a comparison tree $\cA'$
where we simply reject whenever a comparison is inconsistent with the total order, and otherwise mimics $\cA$. 
The comparison tree of $\cA'$ has an error probability at most that of $\cA$ (since it may reject a few $f\neq 2\v$), and is just a path. Hence, it
can be modeled as a non-adaptive distinguisher. We query upfront all the points involving points on this path, and make the relevant comparisons
for the output.
%
%
%Since $\cA$ is deterministic and comparison based, for every $1\leq s\leq t$, for every $\x\in \D^s$, and for every permutation $\sigma:[s]\to[s]$, there exists {\em exactly} one $y\in \D'$ such that $q^y_{\x,\sigma}(\cdot) = 1$. 
%We now describe $\cA'$. Initialize $X$ to $\emptyset$.
%Note that $x_1$, the first point to be queried is predetermined in $\cA$. Add $x_1$ to $X$. At any point $s\geq 1$, let $\pi_s$ denote the permutation of $[s]$ which is the increasing order of $\v(x_1),\ldots,\v(x_s)$. Let $x_{s+1}$ be the $y$ with 
%$q^y_{\x,\pi_s}(\cdot) = 1$. If $x_{s+1} \in \{\acc,\rej\}$, we stop and record this as the answer \ans. Otherwise we add $x_{s+1}$ to $X$. $\cA'$ queries $X$ in advance, and returns \ans if the answers are consistent with the total order of $\{\v(x): x\in X\}$, else returns \rej.
%By definition, $\cA'$ is non-adaptive and comparison based. Since $\cA$ makes at most $t$ queries, so does $\cA'$, since $x_{t+1}\in \{\acc,\rej\}$. Finally, on input $\v$, $\cA'$ behaves similar as $\cA$ (since $\cA'$ returns \ans as well), and on input $g_{j,k}$, it either rejects or returns the same answer as $\cA$. Therefore, the probability of error is at most that of $\cA$.
\end{proof}

\def\v{{\tt val}}

\noindent
Combined with \Prop{comp}, the following lemma completes the proof of \Lem{main-hard}.
% as discussed in the first paragraph by providing a lower
%bound on non-adaptive, deterministic comparison based testers.
\begin{lemma}\label{lem:hard}
Any deterministic, non-adaptive, comparison-based distinguisher $\cA$ making fewer than $t \leq \frac{m'}{8\eps}$ queries, errs with probability at least $1/8$.
\end{lemma}
\begin{proof} %\noindent
%The distribution $\cF_m$ puts probability mass $1/2$ on the function $f=2\v$ and $\frac{\eps}{m'}$ on each of the $g_{j,k}$'s. We are interested in the performance of deterministic comparison based testers on functions drawn from $\cF_m$. 

%First, a simple observation allows us to focus on non-adaptive comparison based testers, for function drawn from $\cF_m$.
%Observe that $\v(x)$ gives a total order on $\{0,1\}^m$. 
%In a possibly adaptive comparison based procedure $\cA$, if any comparison made is 
%inconsistent with the $\v$-ordering on $\{0,1\}^m$, the procedure can correctly 
%reject without making further queries. Hence, the procedure can make all queries in advance, simply assuming that every
%intermediate comparison is consistent with $\v$. 

Let $X$ be the set of points queried by the distinguisher. Set $X_k =: X \cap S_k$; these form a partition of $X$.
%Fix some such $k$, and let $X_k = \{x^{(1)}, x^{(2)}, \ldots, x^{(t)}\}$.
%Observe that the all these points are identical in the largest $m - m'$ bits. Hence,
%all comparisons between $x^{(b)}$ and $x^{(c)}$ only depend on the smallest $m'$ bits.
We say that a pair of points $(x,y)$ \emph{captures} the (unique) coordinate $j$, if $j$ is the largest
coordinate where $x_j \neq y_j$. (By largest coordinate, we refer to the value of the index.)
For a set $Y$ of points, we say $Y$ captures coordinate $j$ if there is a pair in $Y$ that captures $j$.

\begin{claim}\label{clm:viol}
For any $j,k$, if the algorithm distinguishes between $\v$ and $g_{j,k}$, then $X_k$ captures $j$.
%If $\v(x) < \v(y)$ and $g_{j,k}(x) > g_{j,k}(y)$, then $x,y \in S_k$ and $(x,y)$ captures $j$.
\end{claim}
\begin{proof} 
If the algorithm distinguishes between $\v$ and $g_{j,k}$, there must exist $(x,y)\in X$ such that
 $\v(x) < \v(y)$ and $g_{j,k}(x) > g_{j,k}(y)$. We claim that $x$ and $y$ capture $j$; this will also imply they lie in the same $S_{k'}$ since the $m-j$ most significant bit of $x$ and $y$ are the same.

%
%
%Suppose $y \notin S_k$ or $y_j = 0$. Then $g_{j,k}(y) - g_{j,k}(x) \geq 2(\v(y) - \v(x)) > 0$. 
%
%Therefore, $y \in S_k$ and $y_j = 1$. Let $z := \argmin_{w \in S_k} \v(w)$. Observe that the first $m'$ coordinates
%of $z$ are $0$, and $y, z$ agree on all other coordinates. Hence, $g_{j,k}(y) \geq 2\v(z) - 1$.
%
%Suppose $x \in S_{k'}$ where $k' < k$. So $\v(x) < \v(z)$, and $g_{j,k}(x) \leq 2\v(x) < 2\v(z) - 1 \leq g_{j,k}(y)$.
%We are left with the situation that $x, y \in S_k$ and $y_j = 1$.
%If $x_j=y_j$, then $g_{j,k}(y) - g_{j,k}(x) = 2(\v(y) - \v(x))>0$. So $x_j = 0$.

Firstly, observe that we must have $y_j = 1$ and $x_j = 0$; otherwise, $g_{j,k}(y) - g_{j,k}(x) \geq 2(\v(y) - \v(x)) > 0$ contradicting the supposition. Now suppose $(x,y)$ don't capture $j$ implying there exists $i>j$ which is the largest coordinate at which they differ. Since $\v(y)>\v(x)$ we have $y_i=1$ and $x_j=0$. Therefore,
%Let $i$ be the largest coordinate at which $x,y$ differ. Since $\v(x) < \v(y)$, $x_i = 0$ and $y_i = 1$.
%If $i < j$, then $x_j = y_j$, implying $g_{j,k}(y) - g_{j,k}(x) = 2(\v(y) - \v(x)) > 0$ contradicting the supposition.
%If $i > j$, then 
we have $$g_{j,k}(y) - g_{j,k}(x) \geq 2(\v(y)-\v(x)) - 2^{j} - 1
\geq (2^i + 2^j) - \sum_{1 \leq r < i} 2^r - 2^j - 1 > 0.$$
So, $x,y$ capture $j$ and lie in the same $S_{k'}$. If $k'\neq k$, then again $g_{j,k}(y) - g_{j,k}(x) = 2(\v(y)-\v(x)) > 0$. Therefore, $X_k$ captures $j$.
%
%
%
%Since $\v(x) < \v(y)$, $x_i = 0$, $y_i = 1$, 
%and $\v(y) - \v(x) \geq 2^i-1$.
%
%
%We have, $g_{j,k}(y) - g_{j,k}(x) \geq 2(\v(y)-\v(x)) - 2^{j} - 1 =\sum_{\ell:x_\ell > y_\ell} 2^{\ell}-2^j-1 $. 
%If $x$ and $y$ differ in more than one bit (say $i$ and $i'$), then the RHS is at least $2^i -2^j + 2^{i'} - 1 > 0$
%since $i\geq j$ and $i'\geq 1$. Therefore, the only possible violations occur when $x$ and $y$ differ in exactly one bit, and that bit is $j$.
\end{proof}
%I not too hard to see at most $|X_k|-1$ coordinates are captured.

%A pair $(x^{(b)}, x^{(c)})$ is said to \emph{capture} coordinate $i$, if $x^{(b)}$ and $x^{(c)}$ differ only on coordinate $i$.
%A coordinate is captured by $X_k$ if there exists a pair $(x^{(b)}, x^{(c)})$ capturing $i$.
\noindent
The following claim allows us to complete the proof of the lemma.
\begin{claim}\label{clm:cap} A set $Y$ captures at most $|Y|-1$ coordinates.
\end{claim}

\begin{proof} We prove by induction on $|Y|$. When $|Y| = 2$, this is trivially true.
Otherwise, pick the largest coordinate $j$ captured by $Y$ and let $Y_0 =\{y: y_j = 0\}$ and $Y_1 = \{y:y_j = 1\}$. 
By induction, $Y_0$ captures at most $|Y_0|-1$ coordinates, and $Y_1$ captures at most $|Y_1| -1$ coordinates.
Pairs $(x,y)\in Y_0\times Y_1$ only capture coordinate $j$. Therefore, the total number of captured coordinates is at most $|Y_0| - 1 + |Y_1| - 1 + 1 = |Y| - 1$.
\end{proof}

\noindent

%Suppose comparisons in $X$ can distinguish $2\v$ from $g_{j,k}$. By the previous claim, $X_k$ captures $j$. 
If $|X| \leq m'/8\eps$, then there exist at least $1/4\eps$ values of $k$
such that $|X_k| \leq m'/2$. 
By \Clm{cap}, each such $X_k$ captures at most $m'/2$ coordinates.
Therefore, there exist at least 
$\frac{1}{4\eps}\cdot\frac{m'}{2} = \frac{m'}{8\eps}$ functions $g_{j,k}$'s that are indistinguishable from the monotone function $2\v$ to a comparison-based procedure
that queries $X$. 
This implies the distinguisher must err (make a mistake on either these $g_{j,k}$'s or $2\v$) with probability at least $\min(\frac{\eps}{m'}\cdot\frac{m'}{8\eps}, 1/2) = 1/8$.
\end{proof}

\subsection{The final bound}
\def\v{{\tt val}}

Recall, given  function $f:\{0,1\}^{d\ell}\mapsto \NN$, we have 
the function $\f:[n]^d \mapsto \NN$ by defining $\f(\vec{y}) := f(\phi(\vec{y}))$.
We start with the following observation.%a technical fact.

\begin{proposition}\label{prop:dist}
The function $\widetilde{2\v}$ is monotone and
every $\widetilde{g_{j,k}}$ is $\eps/2$-far from being monotone.
\end{proposition}

\begin{proof} Let $\vec{u}$ and $\vec{v}$ be elements in $[n]^d$ such that $\vec{u} \prec \vec{v}$.
We have $\v(\phi(\vec{u})) < \v(\phi(\vec{v}))$, so $\widetilde{2\v}$ is monotone.
For the latter, it suffices to exhibit a matching of violated pairs of cardinality $\eps2^{d\ell}$ for $\widetilde{g_{j,k}}$. This is 
%precisely given by the $j$th dimension cut in the subhypercube $S_k$. 
given by pairs $(\vec{u},\vec{v})$ where $\phi(\vec{u})$ and $\phi(\vec{v})$ only differ in their
$j$th coordinate, and are both contained in $S_k$. Note that these pairs are comparable in $[n]^d$
and are violations.
%Formally, these are the pair $(\phi(\vec{u}), \phi(\vec{v})) \in S_k\times S_k$, $\phi(\vec{u})_j=1,\phi(\vec{v})_j=0$,
%and $\forall i \neq j, \phi(\vec{u})_i = \phi(\vec{u})_i$.
\end{proof}

\begin{theorem}\label{thm:hc}
Any $(t,\eps/2,1/16)$-monotonicity tester for $f:[n]^d\mapsto \NN$, must have $t\geq \frac{d\log n - \log(1/\eps)}{8\eps}$.
\end{theorem}

\begin{proof} By \Thm{fischer}, it suffices to show this for comparison-based $(t,\eps/2,1/8)$ testers.
By Yao's minimax lemma, it suffices to produce a distribution $\cD$ over functions $f:[n]^d\mapsto \NN$ such
that any deterministic comparison-based $(t,\eps/2,1/8)$-monotonicity tester for $\cD$ must have $t\geq s$,
where $s := \frac{d\log n - \log(1/\eps)}{8\eps}$.

Consider the distribution $\cD$ where we generate $f$ from $\cF_{m,\eps}$ and output $\f$. 
Suppose $t < s$. By \Prop{dist}, the deterministic comparison based monotonicity tester acts as a determinisitic comparison-based distinguisher
for $\cF_{m,\eps}$
making fewer than $s$ queries, contradicting \Lem{hard}.
\end{proof}

\section{Conclusion}
In this paper, we exhibit a lower bound of $\Omega(\eps^{-1}d\log n - \eps^{-1}\log \eps^{-1})$ queries on adaptive, two-sided monotonicity testers for functions $f:[n]^d \mapsto \NN$, matching the upper bound of $O(\eps^{-1}d\log n)$ queries of~\cite{ChSe13}. Our proof hinged on two things: that for monotonicity on any partial order one can focus on comparison-based testers, and a lower bound on comparison-based testers for the hypercube domain.  Some natural questions are left open. Can one focus on some restricted class of testers for the Lipschitz property, and more generally, can one prove adaptive, two-sided lower bounds for the Lipschitz property testing on the hypergrid/cube? 
Currently, a $\Omega(d\log n)$-query non-adaptive lower bound is known for the problem~\cite{BlJh+12}.
Can one prove comparison-based lower bounds for monotonicity testing on a general $N$-vertex poset? For the latter problem, there is a $O(\sqrt{N/\eps})$-query non-adaptive tester, and a $\Omega(N^{\frac{1}{\log\log N}})$-query non-adaptive, two-sided error lower bound~\cite{FLNRRS02}.
Our methods do not yield any results for bounded ranges, but there are significant gaps in our understanding for that regime. 
For monotonicity testing of boolean functions $f:\{0,1\}^n \mapsto \{0,1\}$, the best adaptive lower bound
of $\Omega(\log n)$, while the best non-adaptive bound is $\Omega(\sqrt{n})$~\cite{FLNRRS02}. 

\bibliographystyle{alpha}
\bibliography{derivative-testing}

\end{document}